\documentclass[10pt, final, journal, letterpaper, twocolumn]{IEEEtran}
\usepackage{amssymb}
\usepackage{amsmath}
\usepackage{amsfonts}
\usepackage{graphicx}
\usepackage{algorithm}
\usepackage{algorithmic}
\usepackage{epstopdf}
\usepackage{cite}
\usepackage{exscale}
\usepackage{relsize}

\usepackage{color}
\usepackage{multirow}
\usepackage{subfigure}
\usepackage{bm}

\newtheorem{theorem}{\textbf{Theorem}}

\newtheorem{corollary}{\textbf{Corollary}}

\newtheorem{definition}{\textbf{Definition}}

\newtheorem{lemma}{\textbf{Lemma}}

\begin{document}
\title{Modeling and Analysis for Cache-Enabled Networks with Dynamic  Traffic}
\author{ Bin Xia, {\em Senior Member, IEEE}, Chenchen Yang,  Tianyu Cao,\\
\thanks{This work has been accepted by IEEE Communications Letters. Manuscript received June 02, 2016. The associate editor coordinating the review of this letter and approving it for publication was Alexey Vinel.}
\thanks{This work is supported in part by the National Key Science \& Technology Specific Program (Grant No. 2016ZX03001015) and the Shenzhen-Hong Kong Innovative Technology Cooperation Funding (Grant No. SGLH20131009154139588).}
\thanks{The authors are with the Department of Electronic Engineering, Shanghai Jiao Tong University, Shanghai, China. Emails: \texttt {\{bxia, zhanchifeixiang, sherlockdoylecaoty\}@sjtu.edu.cn}. C. Yang is the corresponding author.}
}
\maketitle
\begin{abstract}
Instead of assuming fully loaded cells in the analysis on cache-enabled networks with tools of stochastic geometry, we focus on the dynamic traffic in this letter. With modeling traffic dynamics of request arrivals and departures, probabilities of full-, free-, and modest-load cells in the large-scale cache-enabled network are elaborated based on the traffic queue state. Moreover, we propose to exploit the packets cached  at cache-enabled users as side information to cancel the incoming interference. Then the packet loss rates for both the cache-enabled and cache-untenable users are investigated. The simulation results verify our analysis.
\end{abstract}
\section{Introduction}
Caching is a promising  way to offload the wired backhaul  traffic and the over-the-air traffic, especially for the highly-loaded network \cite{femto2, 5G2, editor}. Poisson Point Process (PPP) of the stochastic geometry is a tractable tool to model and analyze wireless networks. The accuracy and tractability of  PPP are elaborated in \cite{SINR, capacity1}, while the fully loaded cells (full queues at all times) are assumed. Based on this assumption, related studies have been conducted on the cache-enabled network. \cite{geographic} focuses on the caching placement with stochastic geometry.
The outage probability and the delivery rate are investigated with PPP for the backhaul-limited and cache-enabled network in \cite{backhaul1}.
\cite{letter1} explores the  storage-bandwidth tradeoffs with modeling locations of small cell base stations (BSs) as PPP.
However, in previous studies, the assumption of fully loaded cells is too idealistic to fully capture
the randomness of the stochastic network. Traffic dynamics of request arrivals and departures are studied in typical researches such as \cite{editor}.

As BSs transmit wireless signals simultaneously, interference occurs and should be eliminated as much as possible. Cache-induced opportunistic cooperative multiple-input multiple-output (MIMO) is proposed for interference mitigation by caching parts of files at BSs in \cite{Anliu}. Caching split files at transmitters is considered in \cite{aided} to implement interference cancellation. \cite{IA} analyzes the benefits of caching and interference alignment at BSs in MIMO system with limited backhaul.  Cache placement and interference management are elaborated in \cite{tranrecei} where both transmitters and receivers are  cache-enabled. However, backhaul is needed for the inter-connection of cache-enabled transmitters to achieve cooperative transmission in previous studies. Interference cancellation at cache-enabled receivers to relieve backhaul requirement at the transmitter side remains to be studied.

In this letter, dynamic traffic is considered in the stochastic cache-enabled network, and  the modeling method is proposed. Interference cancellation is performed at receivers using the cached packets. The main contributions are as follows:
\begin{itemize}
\item\emph{Stochastic cache-enabled network with dynamic traffic}:  Full-, free-, and modest-load cells are defined and the probabilities are analyzed based on queue states at BSs.
\item\emph{Interference cancellation with cached packets at receivers}: We propose the protocol that cache-enabled users exploit cached packets as side information to cancel incoming interference. The packet loss rates are analyzed.
\end{itemize}
\section{System model}
Consider a large-scale network, where BSs and users are spatially distributed based on two mutually independent PPPs $\Phi_b$ and $\Phi_u$, with intensity $\varphi_b$ and $\varphi_u$, respectively. Dynamic traffic is considered and packet (file) requests of each user  are trigged based on the Poisson process with parameter $\lambda$ $[$packets/second$]$
\cite{editor, Queue}. Users request packets from a given library $\mathcal{C} \triangleq \{c_1, c_2, \cdots ,c_C\}$ and all packets have the same size $T$ $[$Mbits$]$\footnote{For ease of illustration, we assume that all packets have the same size. Results in the letter can be extended to the case of different packet sizes.}.
Each user randomly requests one packet in each time. Different packets are accessed by the user with different frequencies, which can be defined with the  packet popularity distribution $\mathcal{F} \triangleq \{f_1, f_2, \cdots ,f_C\}$ for $f_i\in[0,1]$ and $\sum_{i=1}^C f_i=1$. Without loss of generality, we assume $f_i$ is the popularity of packet $c_i$ and $f_1 \geq f_2,\cdots, \geq f_C$.

We focus on the downlink transmission, and the system is time-slotted with fixed duration $\tau$ $[$seconds$]$. Each BS transmits one packet in a slot with transmission power $P$. A user obtains packets from the BS providing the highest long-term average received signal strength (i.e., from the nearest BS).  Requests are queued in the infinite buffers of the BS until they can be served. The service discipline is FIFO (first-in, first-out) \cite{editor}. Make the assumption that new request arrivals during a slot can not start to be served until the beginning of the  next slot.

Consider a part ($0 <\alpha <1$ proportion) of users are cache-enabled and each of them has a limited caching storage with size of $M\times T$ $[$Mbits$]$. When the network is off-peak, e.g., the nighttime, the $M$ most popular packets are broadcasted via the BSs to and then pre-cached at the cache-enabled users.  Denote the set of the cached packets as $\mathcal{M}=\{c_1,c_2,...,c_M\}$. When  the requested packet is available in the user's local caching set $\mathcal{M}$, the user obtains it immediately; otherwise, the BS transmits the requested packet to the user via the wireless downlink{\footnote{Extra cost (e.g., delay) for the BS to fetch packets from servers via wired backaul is assumed to be neglected. It can be a promising topic for the future.}}.
Denote the cache hit ratio of the cache-enabled user as $\delta=\sum_{i=1}^M f_i$. 
\section{Probabilities of the free-, full- and modest-load cells}
Then  BS cells are polygonal and form the Voronoi tessellation in $\mathbb{R}^2$. The size of the BS cell is a random variable, which can be accurately predicted by the probability density function (PDF) $ f_S(s)\!=\!\frac{(\varphi_b K)^K}{{\Gamma(K)}}{s^{K-1}e^{-\varphi_bKs}},  0\leq s<\infty$. $s$ denotes the cell size, $K=3.575$ is a constant factor and $\Gamma(\cdot)$ is the gamma function. We then have the following lemma \cite{CSMA},
\begin{lemma}
The probability mass function (PMF) of the number of users in the coverage of  a BS is
\begin{align}
&\mathbb{P}(N=n)=\frac{\varphi_{u}^n(K\varphi_b)^K\Gamma(K+n)}{(\varphi_{u}+K\varphi_b)^{K+n}\Gamma(n+1)\Gamma(K)}.
\end{align}
\end{lemma}
\begin{proof}
With the property of PPP and  conditioned on the cell size, the number of users $N$ in a BS cell is a Poisson random variable, which is generated by the conditional probability mass function (PMF),  
$\mathbb{P}(N=n|S=s)={(\varphi_{u}s)^n\text{exp}\{-\varphi_{u}s\}}{(n!)}^{-1}, n=1,2,\ldots$. The PMF of the number of users in a BS cell is $\mathbb{P}(N=n)=\int_{0}^\infty\mathbb{P}(N=n|S=s)f_{S}(s)\mathrm{d}{s}$. Then the proof is finished.
\end{proof}

By noting that $\alpha$ of the users are cache-enabled and $M$ packets have been cached into the cache-enabled users, the packet arrivals at a BS are Poisson process with parameter $\lambda_{b|n}=(1-\alpha\delta)n\lambda$, conditioned on that the number of users in the coverage of  a BS is $N=n$. Denote the number of requests in the queue at a BS after the $k$-th slot as $q_k$, that after the slot $(k+1)$  as $q_{k+1}$.
The total number of request arrivals at a BS during the slot $(k+1)$ is denoted by $a_{k+1}$. So the $a_{k+1}$'s are independent identically distributed Poisson variables with parameter $\lambda_{b|n}$. Further we consider the scheme that a request will be deleted from the buffer no matter whether the requested packet has been transmitted completely at the end of the slot scheduled to it. We then have
$q_{k+1}=\left[q_k-1\right]^++a_{k+1}$. Here, $[x]^{+} \triangleq$ max$\{x,0\}$.  In order for the BS to reach a stochastic equilibrium (to keep steady), the average number of requests arriving in a slot must be lower than the average number of packets that a BS can transmit out \cite{editor, math1}. It implies that the stochastic equilibrium is possible if and only if $\lambda_{b|n}\tau<1$.
It provides the critical condition (i.e., $(1-\alpha\sum_{i=1}^M f_i)\lambda\tau n <1$) for a BS cell to keep steady. Let $\bar\lambda$ denote $(1-\alpha\sum_{i=1}^M f_i)\lambda$, i.e., $\bar\lambda\triangleq(1-\alpha\sum_{i=1}^M f_i)\lambda$. We then have the following definition to clarify cells into three different types,

\textbf{Full-load cell:} Define the cell with the queue such that $\bar\lambda\tau n\geq1$  as full-load cell. Since the BSs out of equilibrium are with infinitely-backlogged packets, we consider that they are transmitting packets at all times due to the saturation condition.

\textbf{Free-load cell:} Define the cell with the queue state such that i) $\bar\lambda\tau n<1$, ii) the queue is empty, as free-load cell  (i.e., there is no request waiting to be served in the queue). The BSs need not transmit packet in the free-load state.

\textbf{Modest-load cell:}  Define the cell with non-empty queue such that $\bar\lambda\tau n<1$ as modest-load cell. The BSs need  transmit packet in the corresponding slot.

Therefore, based on the definition above, the probability that a randomly selected cell is a full-load cell is
\begin{equation}\label{psteady}
p_t=\sum\nolimits_{n=\lceil{(\bar\lambda\tau)}^{-1}\rceil}^\infty\mathbb{P}(N=n),
\end{equation}
where the ceiling function $\lceil x\rceil$ is the smallest integer greater than or equal to $x$. It can be observed that $p_t$ is jointly affected by the physical layer and information-centric parameters, e.g., $\varphi_u,\varphi_b,\lambda, \alpha, f_i$ and $M$. Moreover,  $p_t$ also means the fraction of the full-load cells in the two-dimensional plane.

To further get the probability of the free-load BS cell, we analyze the probability that the queue at the BS of stochastic equilibrium is empty.  Conditioned on the number ($N=n$) of users in the coverage of  a BS, define the conditional probability generating function (PGF) $Q(z|n)$ of the number of requests in the queue at the equilibrious BS (i.e., $n<\lceil{(\bar\lambda\tau)}^{-1}\rceil$)  as
\begin{equation}\label{ztran}
Q(z|n)\triangleq\mathbb{E}\left[z^q|n\right], ~~~ \text{where},~q|n\triangleq\lim_{k\rightarrow\infty}q_{k}|n.
\end{equation}
Here, $\mathbb{E}[x]$ is the expectation over variable $x$. Therefore,
\begin{theorem}
The probability of the free-load BS cells is
\begin{align}\label{zero2}
{p}_{0}=\sum\nolimits_{n=0}^{\lceil{(\bar\lambda\tau)}^{-1}\rceil-1}\mathbb{P}(N=n)(1-\lambda_{b|n} \tau).
\end{align}
\end{theorem}
\begin{proof}
According to the definition in (\ref{ztran}), we have
\begin{align}\label{q}
Q(z|n)&=\lim_{k\rightarrow\infty}\mathbb{E}\left[z^{q_{k+1}}|n\right]\stackrel{(a)}{=}\mathbb{E}\left[z^{a+\left[q-1\right]^+}|n\right]\nonumber\\
&\stackrel{(b)}{=}\mathbb{E}_a\left[z^a|n\right]\mathbb{E}_q\left[z^{\left[q-1\right]^+}\Big|n\right]\triangleq A(z|n)H(z|n).
\end{align}
In Step (a), $a$ is a random variable with the same distribution of $a_{k+1}$. Step (b) is obtained by noting that $a_{k+1}$ and $q_k$ are mutually independent random variables, so the same holds for variables $a$ and $q$. Moreover,
$H(z|n)=\mathbb{E}_q[z^{[q-1]^+}|n]=\sum\nolimits_{i=0}^1\mathbb{P}_{i|n}+\sum\nolimits_{i=2}^\infty\mathbb{P}_{i|n}z^{i-1}=(1-z^{-1})\mathbb{P}_{0|n}+Q(z|n)z^{-1},$
where $\mathbb{P}_{i|n}\triangleq\text{Prob}[q=i|n]$.  Substituting it to (\ref{q}), we have
$Q(z|n)={A(z|n)\mathbb{P}_{0|n}(z-1)}{[z-A(z|n)]^{-1}}.$
Substituting $ Q(z|n)|_{z=1}=\sum_{i=0}^\infty \mathbb{P}_{i|n} z^{i}|_{z=1}=1$ to $Q(z|n)$ and by noting that $A(z|n)=\text{exp}\{-\lambda_{b|n}\tau(1-z)\}$, we have $\mathbb{P}_{0|n}=1-\lambda_{b|n} \tau.$
We then get (\ref{zero2}) with ${p}_{0}=\sum_{n=0}^{\lceil{(\bar\lambda\tau)}^{-1}\rceil-1}\mathbb{P}(N=n)\mathbb{P}_{0|n}$.
\end{proof}


Moreover, based on (\ref{psteady}) and (\ref{zero2}),  the probability of the modest-load BS cells is $p_m=1-p_0-p_t$.
\section{Packet loss rate with and without interference cancellation}
Given the probabilities of BSs in different states, the actually active BSs in each slot are located as thinning PPP with parameter $\varphi_{a}=[p_t+(1-p_t)(1-p_0)]\varphi_b=[1-p_0+p_0p_t]\varphi_b$. Without loss of generality, according to Slivnyak's theorem \cite{SINR, capacity1}, we assume that there is a typical user $u_0$ with or without caching ability at the origin of the Euclidean area.

We first analyze the scenario where the typical user $u_0$ is cache-untenable. The received signal of $u_0$ is given by
 \begin{equation}
 y_0=\sqrt{P}d_{0,0}^{-\frac{\beta}{2}}h_{0,0}x_0+\sum\limits_{k\in\Phi_a\odot b_0 }\sqrt{P}d_{k,0}^{-\frac{\beta}{2}}h_{k,0}x_k+n_0,
 \end{equation}
 where $b_0$ is the serving BS of user $u_0$, $d_{0,0}$ is the distance between $u_0$ and $b_{0}$. $\beta\geq 2$ denotes the path-loss exponent. We consider Rayleigh fading channel from $b_0$ to $u_0$ with average unit power, $|h_{0,0}|^2\sim \exp(1)$. $x_0$ is the transmit signal from $b_0$ to $u_0$ with unit power. $d_{k,0}$ is the distance between interfering BS $k$ and $u_0$. Then  interfering BSs are distributed according to PPP $\Phi_a$ with parameter $\varphi_{a}$ outside the circle which is centered at the origin and with radius $d_{0,0}$ (denoted as $\Phi_a \odot b_0$) . $|h_{k,0}|^2\sim \exp(1)$ is the channel fading between  interfering BS $k$ and $u_0$, $x_k$ is the transmit signal from interfering BS $k$  to user $u_0$ with average unit power, and $n_0\sim \mathcal{CN}(0,\sigma^2)$ is the zero-mean additive white Gaussian noise (AWGN) with power $\sigma^2$. In addition, we consider the channels remain constant in each slot. Then the signal-to-interference-plus-noise ratio (SINR) of the cache-untenable user $u_0$ is
\begin{align}\label{sinr}
    \text{SINR}_u&\!=\!\frac{P|h_{0,0}|^2d_{0,0}^{-\beta}}{\sum_{k\in\Phi_a\odot b_0 }{P|h_{k,0}|^2d_{k,0}^{-\beta}}\!+\!\sigma^2}\!\triangleq\!\frac{P|h_{0,0}|^2d_{0,0}^{-\beta}}{{I}_u+\sigma^2}.
\end{align}
$I_{u}$ is the cumulative interference from the interfering BSs.

Next, we analyze the scenario where the typical user is cache-enabled.  Note that $\alpha$ of users are cache-enabled.  $\delta$ of their requested packets can be obtained immediately from local caching and yet $1-\delta$ of the requested packets need to be obtained from the BSs. On the other hand,  $1-\alpha$ of users are cache-untenable and all of their requests need to be served by the BSs. $\delta$ of requested packets triggered by the cache-untenable users are the same with that in the caching set $\mathcal{M}$. That is, the requested packets of cache-untenable users can be divided into two subsets,  one subset includes the same packets in $\mathcal{M}=\{c_1,c_2,...,c_M\}$, another subset includes the packets in the complementary set $\mathcal{C}\backslash\mathcal{M}=\{c_{M+1},c_{M+2},...,c_C\}$. In other words, on average, if  $\mathcal{P}$ requests are triggered by users in a slot, $\alpha\delta\mathcal{P}$ packets are obtained immediately from the local caching, and that need to be obtained from BSs are given by $[\alpha(1-\delta)+(1-\alpha)]\mathcal{P}=\{\alpha(1-\delta)+[(1-\alpha)(1-\delta)+(1-\alpha)\delta]\}\mathcal{P}=(1-\delta)\mathcal{P}+(1-\alpha)\delta\mathcal{P}
\triangleq\mathcal{P}_{a1}+\mathcal{P}_{a2}$. The $\mathcal{P}_{a2}$ and the  $\mathcal{P}_{a1}$  packets  fall in the caching set $\mathcal{M}$ and the complementary set $\mathcal{C}\backslash\mathcal{M}$, respectively. Then the BSs transmitting the two subsets of packets are respectively distributed according to thinning PPP $\Phi_{a2}$ with intensity $\frac{\mathcal{P}_{a2}}{\mathcal{P}_{a1}+\mathcal{P}_{a2}}\varphi_{a}$ and thinning PPP $\Phi_{a1}$ with intensity $\frac{\mathcal{P}_{a1}}{\mathcal{P}_{a1}+\mathcal{P}_{a2}}\varphi_{a}$. When the cache-enabled user requests an uncached packet, the received signal is given by
\begin{align}
    y_0&=\sqrt{P}d_{0,0}^{-\frac{\beta}{2}}h_{0,0}x_0+\sum\limits_{j\in\Phi_{a1}\odot b_0 }\sqrt{P}d_{j,0}^{-\frac{\beta}{2}}h_{j,0}x_j\nonumber\\
    &+\sum\limits_{k\in\Phi_{a2}\odot b_0 }\sqrt{P}d_{k,0}^{-\frac{\beta}{2}}h_{k,0}x_k+n_0.
 \end{align}
The interferences respectively come from the BSs distributed with $\Phi_{a1}$ and $\Phi_{a2}$ outside the circle which is centered at the origin and with radius $d_{0,0}$ (denoted as $\Phi_{a2} \odot b_0$ and $\Phi_{a1} \odot b_0$).

We consider the ideal case that the channel state information (CSI) knowledge is available at the cache-enabled users. We further assume that, before interfering BSs in $\Phi_{a2}$ transmit packets that have been cached in the cache-enabled users, the cache-enabled users can be informed of the the packet indexes via extra interactive signals. Then the incoming interference from the BSs in $\Phi_{a2}$ can be canceled by  cache-enabled users with the side information (i.e., cached packets in local caching). So the SINR of the typical user is
\begin{align}\label{sinr2}
    \text{SINR}_c&\!=\!\frac{P|h_{0,0}|^2d_{0,0}^{-\beta}}{\sum_{j\in\Phi_{a1}\odot b_0 }{P|h_{j,0}|^2d_{j,0}^{-\beta}}\!+\!\sigma^2}\!\triangleq\!\frac{P|h_{0,0}|^2d_{0,0}^{-\beta}}{I_{c}+\sigma^2}.
\end{align}
$I_{c}$ is the cumulative interference from the interfering BSs.
\begin{definition}
The packet loss rate $\mathcal{P}_l$ is defined as the probability that a packet can not be transmitted completely at the end of the slot scheduled to it, i.e.,
\end{definition}
\begin{align}\label{define}
   \mathcal{P}_l&\triangleq\mathbb{E}\Big[\mathbb{P}\Big[{\tau}B \text{log}_2(1+\text{SINR})< T\Big]\Big]\nonumber\\
   &=\mathbb{E}\Big[\mathbb{P}\Big(\text{SINR}< 2^{\frac{T}{\tau B}}-1\Big)\Big],
\end{align}
where $B$ is the bandwidth. Denote $\bar{T}\triangleq2^{\frac{T}{\tau B}}-1$ hereafter. 
The average is taken over both the spatial PPP and the channel fading.
Firstly, we analyze the packet loss rate (PLR) for the typical user without caching ability. Denote $r$ as the distance between user $u_0$ and its serving node. The PDF of $r$ is $ f_{R}(r)=2\pi \varphi_b r \text{exp}\{-\pi\varphi_b r^2\}$ \cite{SINR}. We then have the following theorem,
\begin{theorem}\label{theorem1}
The PLR of the cache-untenable user is
\begin{align}\label{loss1}
&\mathcal{P}_{l,u}=
1-\nonumber\\
&2\pi \varphi_b\!\!\int_0^\infty\!\!\! r \text{exp}\Big\{\!-\!r^\beta P^{-1}\bar{T}\sigma^2{\setlength\arraycolsep{0.5pt}-}
\pi r^2\left[\varphi_aZ_1(\bar{T})\!+\!\varphi_b\right]\!\!\Big\}\mathrm{d}{r},
\end{align}
where $\mathcal{Z}_1(\bar{T})=\frac{2\bar{T}}{\beta-2}{}_2F_1[1,1-\frac{2}{\beta};2-\frac{2}{\beta};-\bar{T}]$,
and ${}_2F_1[\cdot]$ denotes the Gauss hypergeometric function.
\end{theorem}
\begin{proof}
With the definition in (\ref{define}), we have 
\begin{align}\label{esinr}
\mathcal{P}_{l,u}&=1-\mathbb{E}\left[\mathbb{P}\left(\text{SINR}_u\geq\bar{T}\right)\right]\nonumber\\
                         &=1-\int_0^\infty e^{-r^\beta P^{-1}\bar{T}\sigma^2}\!\mathcal{L}_{I_u}\left[r^\beta P^{-1}\bar{T}|r \right]f_{R}(r)\mathrm{d}{r}.
\end{align}
Here, the interference comes from all of the actually active BSs spatially distributed as PPP $\Phi_a\odot b_0$ with density $\varphi_a$. So the Laplace transform $\mathcal{L}_{I_u}[r^\beta P^{-1}\bar{T}]$ is
\begin{align}\label{laplace1}
&\mathcal{L}_{I_u}\left[r^\beta P^{-1}\bar{T}\right]=\mathbb{E}_{I_u}\left[\text{exp}\left({-r^\beta P^{-1}\bar{T}I_u}\right)\right]\nonumber\\
&=\mathbb{E}_{\Phi_a,\{|h_{k,0}|^2\}}\Big[\text{exp}\Big({-r^\beta P^{-1}\bar{T}{\sum_{k\in\Phi_a\odot b_0 }{P|h_{k,0}|^2d_{k,0}^{-\beta}}}}\Big)\Big]\nonumber\\
&=\text{exp}\Big[{-2\pi\varphi_a\int_{r}^\infty\Big(1-\frac{1}{1+r^\beta\bar{T}v^{-\beta}}\Big)v\mathrm{d}{v}}\Big]\nonumber\\
&=\text{exp}\Big[{-\pi\varphi_ar^2\mathcal{Z}_1(\bar{T})}\Big].
\end{align}
Accordingly, we have (\ref{loss1}) and the proof is finished.
\end{proof}

It can be observed that the packet loss rate is  jointly decided by the  physical layer parameter (e.g., BS density $\varphi_b$ , user density $\varphi_u$, transmit power $P$, the slot duration $\tau$ and pass-loss $\beta$) and the information-centric parameter (e.g., the fraction of cache-enabled users $\alpha$, caching ability $M$, packet popularity $f_i$, request rate $\lambda$ and the packet size $T$).
\begin{corollary}
For the interference-limited network (i.e., $\sigma^2\rightarrow0$) and $\beta=4$, the PLR  is given by
\begin{align}
\!\!\mathcal{P}_{l,u}\!=\!\frac{1}{\Big[(1\!-\!p_0\!+\!p_0p_t)\sqrt{2^{\frac{T}{\tau B}}\!\!-\!\!1}~\!\text{arctan}({\sqrt{2^{\frac{T}{\tau B}}\!\!-\!\!1}})\Big]^{\!-\!1}\!\!\!\!+\!1}.
\end{align}
\end{corollary}

Similarly, for cache-enabled users, we have,
\begin{theorem}\label{theorem2}
When the cache-enabled user access  the uncached packet from the BS, the PLR is
\begin{align}\label{loss2}
\mathcal{P}_{l,c}&=
1-2\pi \varphi_b\int_0^\infty r\text{exp}\Big\{-r^\beta P^{-1}\bar{T}\sigma^2-\pi r^2\nonumber\\
&\times\Big[\frac{\mathcal{P}_{a1}}{\mathcal{P}_{a1}+\mathcal{P}_{a2}}\varphi_{a}Z_1(\bar{T})+\varphi_b\Big]\Big\}\mathrm{d}{r}.
\end{align}
\end{theorem}

\begin{proof}
The derivation can be conducted with referring to the analysis of Theorem \ref{theorem1}, and the proof is omitted here.
\end{proof}
\begin{corollary}
When $\sigma^2\rightarrow0$ and $\beta=4$, the PLR  is,
\begin{align}
\!\!\!\mathcal{P}_{l,c}\!=\!\!\frac{1}{\Big[\frac{(1-p_0+p_0p_t)}{{(1-\alpha\delta)}{(1-\delta)^{-1}}}\sqrt{2^{\frac{T}{\tau B}}\!-\!1}~\!\text{arctan}({\sqrt{2^{\frac{T}{\tau B}}\!\!-\!\!1}})\Big]^{\!-\!1}\!\!\!\!+\!1}.
\end{align}
\end{corollary}

Compare (\ref{loss1}) with (\ref{loss2}), it can be observed that the packet loss rate of the cache-enabled user is lower, owing to the interference cancellation with the cached packets. 
\section{Numerical results and discussions}
The simulation results are obtained with Monte Carlo methods in a square area of $10^4\text{m} \times 10^4\text{m}$, where the nodes are scattered based on PPPs with intensity of   $\{\varphi_u,\varphi_b\}=\{\frac{400}{\pi500^2},\frac{4}{\pi500^2}\}~\text{nodes/m}^2$. The transmit power is $P=43$ dBm and 20 MHz bandwidth are shared among different BSs. We set the path-loss $\beta=4$, the packet arrival rate $\lambda=0.025$ packets/second, the fraction of cache-enabled users $\alpha=0.25$ and the slot duration $\tau=0.5$ second. We consider the packet popularity follows Zipf distribution $f_{i}={\frac{1/i^{\gamma}}{\sum_{j=1}^{C}1/j^{\gamma}}}$, where $\gamma=0.8$ and the total number of packets $C=200$. The size of each packet $T=10$ Mbits and the caching ability $M=10$. Consider the network where caching is not available at users as the baseline, i.e., the fraction of cache-enabled users $\alpha=0$.

The probabilities of free-, full- and modest-load BS cells  are illustrated in Fig. \ref{frac}. The analysis and simulation results are consistent well. Compared with the baseline, the probabilities of full-load cells decreases while that of the free- and modest-load cells increases due to the traffic offloading via caching.

In Fig. \ref{lossm}, it can be observed the packet loss rate of cache-enabled users is far lower than that of the cache-untenable users and the baseline, owing to the interference cancellation gain with caching. When $M=5~(15)$,  the average loss rate of the  cache-enabled network decreases by $9.80\%~(15.46\%)$ compared to that of the baseline. Excluding free-load BSs out of the interfering nodes,  our analysis result for the packet loss rate is lower than that of previous studies  regarding all cells are fully loaded. Our work is a supplement to previous studies.
\begin{figure}[t]
\centering
\includegraphics[width=1.95in]{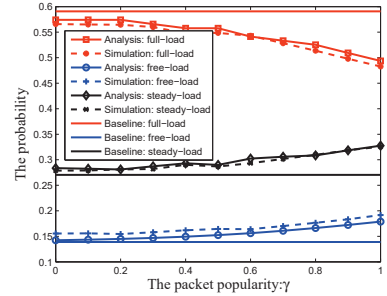}
\caption{The probabilities of BSs in full-, modest- and free-load state for the caching network and the baseline. Note that previous studies consider all cells are active in full-load state, i.e., the probabilities of full-, modest- and free-load BSs are $\{100\%, 0\%, 0\%\}$.}
\label{frac}
\end{figure}
\begin{figure}[t]
\centering
\includegraphics[width=1.95in]{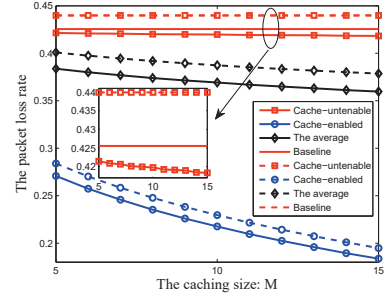}
\caption{The packet loss rate with the change of caching size: solid lines are results of our analysis,  dashed lines are that of previous studies regarding all cells are fully loaded. Note that for previous studies, the packet loss rate of  the baseline overlaps with that of cache-untenable users \cite{SINR, capacity1}.}
\label{lossm}
\end{figure}
\bibliographystyle{IEEEtran}
\bibliography{paperCTY}
\end{document}